\pdfoutput=1
\documentclass[11pt,twocolumn]{article}
\ifdefined\pdfgentounicode
  \input{glyphtounicode}
\fi
\usepackage[T1]{fontenc}
\usepackage[utf8]{inputenc}
\usepackage{amsmath,amssymb,amsthm}
\usepackage{newtxtext}
\usepackage{geometry}
\usepackage{microtype,graphicx,xcolor,booktabs,array,tabularx}
\usepackage{algorithm,algpseudocode,listings}
\usepackage{etoolbox}
\usepackage[nospread]{flushend}
\usepackage{placeins,flafter}
\usepackage[varqu,scaled=0.95]{zi4}
\usepackage{tikz}
\usetikzlibrary{arrows.meta,positioning}
\usepackage{caption}
\definecolor{linkcol}{HTML}{af1834}
\definecolor{codebar}{HTML}{4063D8}
\definecolor{codebg}{HTML}{F6F7F9}
\definecolor{jlkw}{HTML}{9558B2}
\definecolor{jlfun}{HTML}{1F5FAD}
\definecolor{jlcom}{HTML}{626A73}
\usepackage{cite}
\mathchardef\citepunctpenalty=10000 
\usepackage{hyperref}
\usepackage{xurl}
\hypersetup{colorlinks=true,linkcolor=linkcol,citecolor=linkcol,urlcolor=linkcol,
 pdftitle={QRecoupling.jl: numerical and exact computation of classical and quantum recoupling coefficients},
 pdfauthor={Seth K. Asante}}
\usepackage[nameinlink,capitalise]{cleveref}
\makeatletter
\providecommand{\theHALG@line}{}
\renewcommand{\theHALG@line}{\thealgorithm.\arabic{ALG@line}}
\makeatother
\theoremstyle{plain}
\newtheorem{proposition}{Proposition}
\theoremstyle{remark}
\newtheorem{remark}[proposition]{Remark}

\algnewcommand{\IfThen}[2]{\State\algorithmicif\ #1\ \algorithmicthen\ #2}
\newcommand{\q}{\quad}
\newcommand{\qq}{\qquad}
\newcommand{\Q}{\mathbb{Q}}
\newcommand{\Z}{\mathbb{Z}}
\newcommand{\sixj}[6]{\begin{Bmatrix}#1&#2&#3\\#4&#5&#6\end{Bmatrix}}

\DeclareRobustCommand{\code}[1]{{\ttfamily\bfseries\small\detokenize{#1}}}
\DeclareTextFontCommand{\texttt}{\ttfamily\bfseries}
\newenvironment{qrwidefigure}{\begin{figure*}[t]}{\end{figure*}}
\newenvironment{qrwidetable}{\begin{table*}[!t]}{\end{table*}}
\lstdefinelanguage{julia}{
 morekeywords={using,for,in,if,else,end,let,do,begin,function,return,true,false},
 morekeywords=[2]{q6j,q3j,qcg,qdim,fsymbol,fmatrix,radical,xform,qeval,prove_identity,Fsymbol},
 morekeywords=[3]{Exact,Level,At,Symbolic,QForm,XValue,FactorialSum,BigFloat,Nemo,TensorKit,SU2k,FibonacciAnyon},
 sensitive=true,morecomment=[l]{\#},morestring=[b]"}
\BeforeBeginEnvironment{lstlisting}{\par\noindent\begin{minipage}{\linewidth}}
\AfterEndEnvironment{lstlisting}{\end{minipage}\par\medskip}
\title{QRecoupling.jl: numerical and exact computation\\of classical and quantum recoupling coefficients}
\author{Seth K. Asante\thanks{Email: \href{mailto:seth.asante@unb.ca}{seth.asante@unb.ca}}\\[4pt]
  \small Department of Mathematics and Statistics, University of New Brunswick\\
  \small Fredericton, NB E3B 5A3, Canada}
\date{}

\renewenvironment{abstract}
  {\par\small\noindent\textbf{\abstractname.}\ \ignorespaces}
  {\par}
\makeatletter
\apptocmd{\@maketitle}{\begin{center}\begin{minipage}{0.94\textwidth}
\begin{abstract}
QRecoupling.jl evaluates classical and quantum recoupling coefficients, such as Clebsch--Gordan coefficients and $3j$ and $6j$ symbols, and user-defined sums of factorial products, with numerical, exact and symbolic targets. One compact factorial rule covers classical and deformed evaluations, while its cyclotomic exponents reveal common factors and determine which terms contribute at unitary levels. Adaptive arithmetic resolves strongly cancelling sums, while exact tests in quantum integers distinguish algebraic zeros from small nonzero values without a computer-algebra system. At spin cutoff 500, classical $3j$ and $6j$ symbols are evaluated about 7 and 15 times faster than with WIGXJPF. For the tested quantum deformations, relative errors remain below $4\times10^{-14}$ through spin cutoff 400. Recurrences compute fusion matrices with 401 channels about 100 times faster than loops over their entries. The resulting coefficients support recoupling identities, changes of coupling basis and tensor-network calculations.
\end{abstract}\medskip
\end{minipage}\end{center}\vspace{1em}}{}{\PackageError{qrecoupling}{Cannot extend title block}{}}
\makeatother

\begin{document}
\raggedbottom
\maketitle

\section*{Program summary}
\begingroup
\small
\setlength{\parindent}{0pt}
\setlength{\parskip}{3pt plus 1pt minus 1pt}

\textit{Program title:} \href{https://github.com/sethkasante/QRecoupling.jl}{QRecoupling.jl}, version 0.6.4.

\textit{Licensing provisions:} MIT.

\textit{Programming language:} Julia (1.10 or later).

\textit{External routines/libraries:} LRUCache.jl, PrecompileTools.jl. Nemo.jl and TensorKitSectors.jl are optional.

\textit{Nature of problem:} Finite factorial sums for recoupling coefficients can overflow or lose accuracy through cancellation. At roots of unity, vanishing factors also require exact cancellation before specialization.

\textit{Solution method:} A compact factorial rule supports numerical, exact and symbolic evaluation classically, at unitary level $k$ and at real or complex $q$. Scaling and adaptive precision control range and cancellation, while recurrences accelerate families and fusion matrices. A compensated pass reduces rounding error before higher precision is needed. Integer and modular tests distinguish exact zeros from small values. The factorial rule retains an exact description without polynomial expansion. Optional extensions provide polynomial and radical forms through Nemo, and recoupling data for TensorKit.

\textit{Additional comments and restrictions:} Labels must satisfy admissibility rules. Exact weighted coupling forms support equality and zero tests but not arithmetic. Radical arithmetic is limited, and numerical error estimates depend on the arithmetic model and floating-point range.

\par
\endgroup

\newpage 

\section{Introduction}\label{sec:intro}

Many quantities in the quantum theory of angular momentum and its deformations are finite sums of products and quotients of factorials with alternating signs. Examples are Clebsch--Gordan coefficients, $3j$ and $6j$ symbols, their quantum-group analogues, and certain terminating hypergeometric series. Such formulas are difficult to evaluate reliably and exactly. Their terms can exceed the sum by many orders of magnitude, so that floating-point addition may require high precision to resolve it. These terms can also cancel exactly (a \emph{cancellation zero}), which increasing floating-point precision alone cannot certify.\footnote{The quantum $6j$ symbol $\{1\,\tfrac52\,\tfrac52;\tfrac52\,3\,3\}_q$ at $q=e^{i\pi/12}$ vanishes through cancellation of its terms.} Quantum factorials at a root of unity add a third difficulty, because individual factors vanish there and common factors must be cancelled before the formula is specialized.

These sums are the local coupling and recoupling coefficients of many physical theories. For instance, in atomic and nuclear spectroscopy they combine angular momenta~\cite{edmonds}. In quantum gravity the $6j$ symbol is the amplitude of a quantum tetrahedron~\cite{pr}, the building block of transition amplitudes~\cite{rovelli,baezbarrett}, and its quantum-group deformation is used to describe a cosmological constant and curved discrete geometries~\cite{mizoguchi,majorsmolin,hhkr,dupuisgirelli,dgl,livine,hhp}. At a root of unity the coefficients become the associativity data of $\mathrm{SU}(2)_k$, which define the Turaev--Viro and Reshetikhin--Turaev topological quantum field theories~\cite{tv,kr,rt} and the $F$-symbols of anyon models and string-net condensates~\cite{kitaev,levinwen}. Spin foam computations~\cite{dona,gozzini} and tensor-network algorithms with $\mathrm{SU}(2)$ or anyonic symmetry~\cite{singhvidal,pfeifer,asante_tn} evaluate them in very large numbers. These applications need fast approximate values, and exact ones for coherence identities such as the pentagon equation, often within the same calculation.

Existing methods for the classical coefficients exploit different aspects of the factorial structure. WIGXJPF uses prime factorization and multiword integer arithmetic to evaluate classical Wigner symbols without loss from cancellation~\cite{wig}. WignerSymbols.jl\footnote{\texttt{WignerSymbols.jl}, version 2.0, \url{https://github.com/Jutho/WignerSymbols.jl}.} builds on this approach for exact classical output. Other classical approaches compute the coefficients exactly~\cite{wei} or store precomputed values in compact tables~\cite{raschyu}, and recurrence methods are especially useful for families of related values~\cite{kr,sg,lulu,tw}. For quantum-deformed coefficients there are far fewer methods, and the prime-factorization approach does not carry over directly, because quantum integers factor into cyclotomic polynomials rather than primes~\cite{dcr}, and some of these factors vanish at roots of unity. QWignerSymbols.jl\footnote{\texttt{QWignerSymbols.jl}, version 0.1, \url{https://github.com/QuantumKitHub/QWignerSymbols.jl}.} implements the factorial formulas by direct floating-point summation and integrates these data with TensorKit~\cite{tensorkit}. Its tested version does not provide adaptive precision or dedicated exact, symbolic or unitary-level evaluation targets. These methods leave a need for a common interface that handles cancellation, exact zero decisions and root-of-unity specialization across the classical and deformed regimes.

{QRecoupling} addresses the three difficulties through one interface and one set of conventions. It covers Clebsch--Gordan coefficients, $3j$ and $6j$ symbols, fusion and braiding data, and user-defined sums in the classical limit, at principal unitary levels and at real or complex deformation. An optional extension also supplies TensorKit sectors for the supported symmetries. Underlying these calculations, each coefficient is stored as a compact \emph{factorial rule}, to which quantum coupling coefficients add powers of $q$. The $6j$ symbol serves as the main worked example, illustrating numerical, exact and symbolic evaluation from the same rule.

The package develops the deferred cyclotomic framework of Ref.~\cite{dcr} into a compact representation based on factorial rules. This representation separates cancellation within a product from cancellation between terms. Within each product, cyclotomic exponents identify common factors, zeros and poles through integer divisions and determine when the product is a polynomial. Adaptive arithmetic then controls the loss of accuracy in the sum, while recurrences provide an alternative route for supported families and exact tests settle unresolved zeros. Throughout these operations, the rule itself remains a compact exact description, while polynomial expansion and conversion to nested square roots (see \cref{prop:radicals}) are optional operations for calculations that benefit from these explicit forms. 

The article is organized as follows. \Cref{sec:firstcalculation} previews the user interface. \Cref{sec:architecture} writes coupling, recoupling and fusion coefficients as finite factorial sums. \Cref{sec:multiplicative} introduces the factorial rule and its cyclotomic exponents, the main mathematical objects, which determine when terms are integer polynomials and which terms contribute at a level. \Cref{sec:numerics} treats adaptive numerical evaluation with its error estimates and zero decisions, and \cref{sec:exact} exact and symbolic computation. \Cref{sec:workflow} gives worked examples, \cref{sec:performance} reports accuracy and performance, and \cref{sec:limits} discusses scope and future work. The appendix derives the error estimate of the plain pass.

\section{A unified interface}\label{sec:firstcalculation}

\begin{table*}
\centering\small
\caption{The target interface for factorial rules. An omitted deformation is classical. The exact level and symbolic forms of the weighted coupling functions \code{qcg} and \code{q3j} carry powers of $q$, and support exact equality and zero tests but not yet arithmetic.}\label{tab:targets}
\begin{tabularx}{\textwidth}{@{}l l X@{}}
\toprule
Request & Target & Output form \\
\midrule
Classical numerical & \code{Classical()} or omitted & \texttt{Float64} (at $q=1$)\\
Classical exact & \code{Exact()} & $\sigma\sqrt{r}$ with $\sigma=\pm1$ and $r$ rational, and exact sums of such values\\
Unitary level & \code{Level(k)} & \texttt{Float64} (at $q=e^{i\pi/(k+2)}$)\\
Exact level & \code{Exact(k)} & Expression in quantum integers at $q=e^{i\pi/(k+2)}$, with exact arithmetic\\
Specified deformation & \code{At(q)} & \texttt{Float64} or \texttt{ComplexF64} \\
Generic expression & \code{Symbolic()} & Form in quantum integers, a function of $q$\\
\bottomrule
\end{tabularx}

\end{table*}

A calculation specifies integer or half-integer spin labels together with an evaluation target (\cref{tab:targets}), which selects the arithmetic while preserving the convention of the coefficient. The functions \code{qcg}, \code{q3j} and \code{q6j} then evaluate Clebsch--Gordan coefficients, $3j$ symbols and $6j$ symbols, respectively, as illustrated below.

\begin{lstlisting}
using QRecoupling
js = (1, 1, 1, 1, 1, 1)             # spin labels
c = q6j(js...)                      # classical
v = q6j(Exact(10), js...)           # exact level
s = q6j(Symbolic(), js...)          # formula in q
cg = qcg(At(0.8), 1, 1, 1, -1, 1, 0)
w = q3j(At(0.8), 1, 1, 1, 1, -1, 0)
@assert c ≈ 1/6
@assert cg ≈ sqrt(qdim(At(0.8), 1)) * w
@assert numeric_value(v) ≈ q6j(Level(10), js...)
\end{lstlisting}

The coupling calls take spin and magnetic labels, arranged in \code{qcg} as the pairs of $\langle j_1m_1;j_2m_2|jm\rangle$ and in \code{q3j} as the two rows of the $3j$ symbol. The assertion linking the two outputs checks the normalization in \eqref{eq:threejcg}. For a principal unitary root, \code{Level(k)} retains the exact level and its fusion support, whereas \code{At(q)} evaluates the number supplied, even when it is a nearby floating-point approximation to that root.

Exact level values and symbolic results are written in quantum integers, with other representations requested explicitly. For numerical use, \code{numeric_value} converts exact data to a number, while a symbolic result \code{s} is evaluated at a chosen $q$ as \code{s(q)}. When an expanded form is needed, the optional Nemo backend provides \code{xform} to expand a value in $x=q+q^{-1}$ and \code{radical} to look for an expression in nested square roots. The same targets also apply to user-defined rules through \code{qeval}, as illustrated in \cref{sec:userdefined}.

\paragraph*{What a result means.} The targets return three kinds of answer with different guarantees. For unweighted level evaluation, the plain and compensated passes accept a floating-point value when its estimated absolute error is at most $2^{-40}$ or $2^{-44}$ times its magnitude, respectively (\cref{sec:escalation}). Rejected candidates trigger further evaluation stages, and other targets use the acceptance checks described in \cref{sec:complex}.
An exact value is an algebraic number whose equality is decided without tolerance. A symbolic result represents a function of $q$. Equality means that two such functions agree wherever both are regular, using the factorwise principal-square-root convention specified in \cref{sec:complex}. A zero is never inferred from a small or rounded number, i.e., a cancellation zero is returned only when a structural identity or exact arithmetic proves it (\cref{sec:zerodecisions}). Thus, an unproved zero can appear only through underflow of the output type.

\section{Coefficients as finite factorial sums}\label{sec:architecture}
Finite sums of factorial products provide the common computational object considered here. Coupling coefficients add selection rules, phase weights and square-root normalizations.

\subsection{Quantum integers and factorials}
For $q\ne0$, define the symmetric quantum integer and factorial by
\begin{equation}\label{eq:qint}
 [n]_q=\frac{q^n-q^{-n}}{q-q^{-1}},\qq
 [n]_q!=\prod_{r=1}^n[r]_q,\qq [0]_q!=1.
\end{equation}
We write $[n]$ and $[n]!$ when the value of $q$ is clear from the context. The removable singularities at $q=\pm1$ are understood by continuation. At $q=1$, the continuous limit gives $[n]_1=n$ and $[n]_1!=n!$. The principal unitary level $k\ge0$ corresponds to the root of unity
\begin{equation}\label{eq:level}
 h=k+2,\qq q=e^{i\pi/h},\qq
 [n]_q=\frac{\sin(n\pi/h)}{\sin(\pi/h)}.
\end{equation}

\subsection{Finite factorial-product sums}
The computational object is a finite sum of the form
\begin{equation}\label{eq:generalrule}
\begin{aligned}
\mathcal S_q&=\sigma R_0(q)^\eta \sum_{z=z_0}^{z_1}(-1)^{\epsilon z}M_z(q),\\
M_z(q)&=\prod_{r=1}^{m}[a_rz+b_r]_q!^{c_r},
\end{aligned}
\end{equation}
where $R_0(q)$ is a product of quantum factorials with integer exponents, $\eta\in\{1,\tfrac12\}$, $\epsilon\in\{0,1\}$ and $\sigma\in\{-1,0,1\}$. The coefficients $a_r,b_r,c_r$ and the endpoints $z_0,z_1$ are integers, with nonnegative factorial arguments throughout the support. Negative exponents represent reciprocal factors. These discrete data form a \code{FactorialSum}, and the target selects the arithmetic that evaluates it. A square-root prefactor also needs the branch convention of that target.

A classical example is Gauss' hypergeometric function 
\begin{equation}\label{eq:hypergeometricexample}
 {}_2F_1(-n,a;b;1)=\sum_{z=0}^{n}\frac{(-n)_z(a)_z}{(b)_z\,z!},
\end{equation}
where $(a)_z=\prod_{r=0}^{z-1}(a+r)$ and $(a)_0=1$, for integers $n\ge0$ and $a,b\ge1$. Its summand is an alternating factorial ratio with nonnegative integer arguments, and replacing those factorials by symmetric quantum factorials gives a rule of the form \eqref{eq:generalrule}, which \cref{sec:userdefined} implements. The construction extends to other terminating generalized hypergeometric series with suitable integer parameters. Standard basic hypergeometric conventions can require additional powers of $q$.

\subsection{Coupling, recoupling and fusion matrices}

Spins $j\in\tfrac12\Z_{\ge0}$ label the irreducible representations of $\mathrm{SU}(2)$ and its quantum deformation $U_q(\mathfrak{sl}_2)$~\cite{drinfeld,jimbo,kr,kassel,kitaev}. Writing $V_j$ for the representation of spin $j$, a triple $(a,b,c)$ is \emph{admissible} when $V_c$ lies in the decomposition of $V_a\otimes V_b$ into irreducibles, that is, when it satisfies the triangle inequalities and in addition $a+b+c \in \Z$. At a unitary level $k$ the fusion rules of $\mathrm{SU}(2)_k$ further require $a+b+c\le k$~\cite{kitaev}. Coupling and recoupling coefficients with inadmissible labels are returned as zero.

Clebsch--Gordan coefficients transform product states $|j_1m_1\rangle\otimes|j_2m_2\rangle$ into coupled states $|jm\rangle$, with $m=m_1+m_2$. The $3j$ symbols express the same coupling in a different normalization~\cite{edmonds,varshalovich,kr}. The functions \code{qcg} and \code{q3j} implement these coefficients following a common $U_q(\mathfrak{sl}_2)$ convention~\cite{kr}. Their finite sums include monomial weights, taking the form
\begin{equation}\label{eq:weightedrule}
\mathcal C_q=q^{e/2}\sigma R_0(q)^{1/2}
      \sum_{z=z_0}^{z_1}(-1)^zq^{wz}M_z(q),
\end{equation}
whose factorial data are those of the classical formula. For the action $K|m\rangle=q^m|m\rangle$ and the coproduct $\Delta(E)=E\otimes K+K^{-1}\otimes E$~\cite{kr}, the weights of $\langle j_1m_1;j_2m_2|jm\rangle_q$ are the integers
\begin{equation}\label{eq:cgweights}
\begin{aligned}
 w&=-(j_1+j_2+j+1),\\
 e&=(j_1+j_2-j)(j_1+j_2+j+1)\\[-2pt]
  &\q+2(j_1m_2-j_2m_1),
\end{aligned}
\end{equation}
stored beside the rule. The $3j$ symbol is the same coefficient in another normalization,
\begin{equation}\label{eq:threejcg}
\begin{aligned}
 \begin{pmatrix}j_1&j_2&j_3\\m_1&m_2&m_3\end{pmatrix}_q
 &=\frac{(-1)^{j_1-j_2-m_3}}{\sqrt{[2j_3+1]}}\\[-2pt]
 &\q\times\langle j_1m_1;j_2m_2|j_3,-m_3\rangle_q .
\end{aligned}
\end{equation}
The weights are essential away from $q=1$ because substituting quantum factorials into the classical formula alone gives coefficients that are not orthogonal. 

The $6j$ symbols relate two ways of coupling three spins and are evaluated by \code{q6j}. Their dimension-weighted entries form the \emph{fusion} or \emph{associativity matrix} $F^{abc}_d$ between intermediate channels $e$ and $f$. These matrices implement the $F$-moves of anyon models, fusion categories and tensor networks~\cite{kitaev,levinwen}. We use the convention
\begin{equation}\label{eq:fmatrix}
\begin{aligned}
\bigl[F^{abc}_d\bigr]_{ef}
&=(-1)^{a+b+c+d}\sqrt{[2e+1][2f+1]}\\[-2pt]
&\hspace{14mm}\times\sixj{a}{b}{e}{c}{d}{f}_q,
\end{aligned}
\end{equation}
returned as a whole matrix with its channel labels by \code{fmatrix} and entrywise by \code{fsymbol}. It is real orthogonal classically, for positive real $q$, and at principal unitary levels. At nonsingular complex $q$, with consistent square-root branches, the identity $F^{\mathsf T}F=I$ uses the transpose rather than the Hermitian adjoint, so the matrix is complex orthogonal but not unitary. 

The coupling coefficients and the $F$-matrix are tied together by the coassociativity of the coproduct. Coupling three spins in the two possible orders gives two bases, and $F$ is the change of basis between them. With the conventions \eqref{eq:cgweights} and \eqref{eq:fmatrix} this relation holds without an additional phase, and it can be verified within the package from \code{qcg} and \code{fsymbol}.

Other related data are supported from the same interface, such as braiding eigenvalues (\code{rmatrix}), modular $S$ and $T$ matrices, and the tetrahedrally symmetric $G$-symbol of the Turaev--Viro state sum~\cite{tv}. An optional extension provides the TensorKit sector types \code{SU2k(k)} for a level and \code{SU2q(q)} for real $q>0$, whose $F$-symbols, and $R$-symbols at a level, come from this package.

\subsection{The Racah formula}\label{sec:racah}

The factorial rule of the quantum $6j$ symbol with spin labels ${\bf j}=(j_1,\ldots,j_6)$ is determined by four triangle sums and three quadrangle sums,
\begin{equation}\label{eq:ab}
\begin{aligned}
\alpha={}&(j_1+j_2+j_3,\ j_1+j_5+j_6,\\[-2pt]
         &\phantom{(}j_2+j_4+j_6,\ j_3+j_4+j_5),\\
\beta={}&(j_1+j_2+j_4+j_5,\ j_1+j_3+j_4+j_6,\\[-2pt]
        &\phantom{(}j_2+j_3+j_5+j_6),
\end{aligned}
\end{equation}
satisfying the balance $\sum_i\alpha_i=\sum_r\beta_r$. Define the squared triangle normalization by
\begin{equation}\label{eq:triangle}
D(a,b,c)=\frac{[a+b-c]_q![a-b+c]_q![-a+b+c]_q!}{[a+b+c+1]_q!}.
\end{equation}
Writing $\mathcal{N}^2$ for the product of $D$ over the four triples whose sums are the $\alpha_i$, the Racah formula is~\mbox{\cite{kr,kl,cfs}}
\begin{equation}\label{eq:racah}
\begin{aligned}
\mathcal W_q({\bf j})&=\sixj{j_1}{j_2}{j_3}{j_4}{j_5}{j_6}_q= {\cal N}  \sum_{z=\max\alpha}^{\min\beta}(-1)^z\, T(z) ,\\
T(z)&=\frac{[z+1]_q!}{\prod_{i=1}^4[z-\alpha_i]_q!\prod_{r=1}^3[\beta_r-z]_q!}.
\end{aligned}
\end{equation}
Thus \eqref{eq:generalrule} applies with $R_0={\cal N}^2$, $\eta=\tfrac12$, $\sigma=1$ and $\epsilon=1$. At $q=1$, this gives the classical $6j$ symbol whose summation is a terminating ${}_4F_3(1)$ series. The dimension factors in \eqref{eq:fmatrix} are further factorial ratios, so its entries belong to the same class. The positive square root defines $\cal N$ classically, for positive real $q$ and on admissible unitary support, while complex branches are specified in \cref{sec:complex}.

\section{Factorial rules and multiplicative structure}\label{sec:multiplicative}
A \emph{factorial rule} is the list of integers that specifies one sum of the form \eqref{eq:generalrule}. It holds the formula as data, with no value of $q$ and no evaluated number, and every evaluation route starts from it. From these data, a rule supplies both the term ratios that advance the sum and the factor exponents that show which factors cancel within a product and which terms survive at a level.

\begin{algorithm}[t]
\caption{Evaluation of a factorial rule on a regular segment}\label{alg:factorialrule}
\begin{algorithmic}[1]
\Require rule $\mathcal R$, regular segment $[z_0,z_1]$, target
\Ensure $\sigma R_0(q)^\eta\sum_{z=z_0}^{z_1}t_z$
\IfThen{$\sigma=0$ or $z_0>z_1$}{\Return $0$}
\State $w_0\gets\sigma R_0(q)^\eta\,t_{z_0}$ \Comment{one factorial product}
\State $H\gets1$
\For{$z=z_1-1,\ldots,z_0$}
  \State $r_z\gets t_{z+1}/t_z$ \Comment{quantum integers}
  \State $H\gets1+r_zH$
\EndFor
\State \Return $w_0H$
\end{algorithmic}
\end{algorithm}

\subsection{The factorial rule}\label{sec:storage}
In the notation of \eqref{eq:generalrule}, a factorial rule is the tuple
\begin{equation}\label{eq:ruledata}
\mathcal R=\bigl(\sigma,\ \{(n,c)\},\ \{(a,b,c)\},\ z_0,\ z_1,\ \eta,\ \epsilon\bigr).
\end{equation}
The pairs $(n,c)$ stand for the factors $[n]!^c$ of the prefactor $R_0$, and the triples $(a,b,c)$ for the factors $[az+b]!^c$ of the summand $M_z$. The remaining entries are the sign ($\sigma$), the summation range ($z_0:z_1$), and the two flags for a square-root prefactor ($\eta$) and an alternating sign ($\epsilon$). The type \code{FactorialSum} stores exactly these data, and a weighted coupling coefficient adds the two integers $w$ and $e$ of \eqref{eq:weightedrule}. For a fixed formula, the number of stored integers is independent of the labels and of the length of the sum, although the integers themselves grow with the labels. For the $6j$ symbol, \eqref{eq:triangle} and \eqref{eq:racah} give
\begin{equation}\label{eq:prefab}
{\cal N}^2=\frac{\prod_{i=1}^4\prod_{r=1}^3[\beta_r-\alpha_i]!} {\prod_{i=1}^4[\alpha_i+1]!},
\end{equation}
so its rule has sixteen prefactor factorials and eight summand factorials at every spin, all determined by the lists $\alpha$ and $\beta$. Independent permutations of $\alpha$ and $\beta$ generate the $4!\cdot3!=144$ symmetries of the $6j$ symbol, comprising the 24 tetrahedral symmetries and the Regge symmetries~\cite{kr,regge}. The implementation sorts both lists, so equivalent labellings use identical factorial data without a separate canonicalization step.

Once the rule is constructed, evaluation proceeds through ratios of consecutive terms in the sum~\eqref{eq:generalrule}. This is standard for hypergeometric series, and it is especially efficient for factorial products because the ratio of two consecutive terms is a product of quantum integers. For the Racah terms $t_z=(-1)^zT(z)$,
\begin{equation}\label{eq:ratio}
r_z = \frac{t_{z+1}}{t_z} =-\frac{[z+2]\prod_{r=1}^3[\beta_r-z]}{\prod_{i=1}^4[z+1-\alpha_i]}.
\end{equation}
Each step involves eight quantum integers. More generally, write $t_z=(-1)^{\epsilon z}M_z(q)$ for the terms of \eqref{eq:generalrule}. Cancelling factorials in $r_z=t_{z+1}/t_z$ leaves a product of quantum integers and their reciprocals, read directly from the stored affine arguments. Since $t_{z+1}=r_zt_z$, the sum over a segment can be written in nested form,
\begin{equation}\label{eq:horner}
\sum_{z=z_0}^{z_1}t_z=t_{z_0}\Bigl(1+r_{z_0}\bigl(1+r_{z_0+1}(1+\cdots)\bigr)\Bigr),
\end{equation}
where the innermost bracket is $1+r_{z_1-1}$. This is Horner's method~\cite{higham}, applied to the ratios, with \cref{alg:factorialrule} evaluating the brackets from the innermost outwards. Only the first term is formed as a factorial product, and every further term costs one ratio, one multiplication and one addition.

A segment is \emph{regular} when the prefactor and every term in it are finite and nonzero, so that all ratios are defined. At a root of unity the regular segments are found from the valuations of \cref{sec:cyclo}. Each evaluation kernel is constructed from this template (\cref{fig:architecture}). The numerical kernels accumulate the sum with scaling, forward or by compensated Horner nesting, under precision control, while the exact kernels clear the denominators of the ratios and work with Julia's own integers and with residues modulo primes (\cref{sec:numerics,sec:exact}). Julia's multiple dispatch connects each target to its kernel~\cite{julia}, and the optional Nemo backend~\cite{nemo} supplies expansions in $x$ and nested radicals.

\begin{qrwidefigure}
\centering
\begin{tikzpicture}[
 box/.style={draw=black!35,rounded corners=1.5pt,align=center,inner sep=5pt,font=\small},
 ar/.style={-{Stealth[length=4pt]},thick,black!55},
 qa/.style={-{Stealth[length=4pt]},thick,dashed,linkcol!70}]
\node[box] (in) {spin labels or a finite factorial-product sum};
\node[box,fill=linkcol!5,draw=linkcol,below=6mm of in] (rule)
 {\textbf{factorial rule}\\\scriptsize tuple of integers, exponents, booleans\\\scriptsize $\mathcal R=\bigl(\sigma,\ \{(n,c)\},\ \{(a,b,c)\},\ z_0,\ z_1,\ \eta,\ \epsilon\bigr)$};
\node[box,right=15mm of rule,align=left] (val)
 {cyclotomic valuations\\\scriptsize $E_h$: zeros, poles, contributing range\\
  \scriptsize $E_d$: prefactor cancellation, square roots};
\node[box,below=12mm of rule] (ex) {exact arithmetic\\\scriptsize integer recursions and modular residues};
\node[box,left=6mm of ex] (nu) {numerical evaluation\\\scriptsize scaled ratios and precision control};
\node[box,right=6mm of ex] (sy) {exact and symbolic forms\\\scriptsize quantum integers, expansion on request};
\draw[ar] (in)--(rule);
\draw[qa] (rule)--(val);
\draw[ar] (rule)--(ex);
\draw[ar] (rule.south)--++(0,-4mm)-|(nu.north);
\draw[ar] (rule.south)--++(0,-4mm)-|(sy.north);
\node[box,below=7mm of nu] (fa) {related coefficients\\\scriptsize shared tables and column recurrences};
\draw[ar] (nu)--(fa);
\end{tikzpicture}
\caption{The factorial rule connects numerical, exact and symbolic evaluation. Ratios advance the sum, and exact and symbolic forms retain the rule, which is written in quantum integers when it is shown. Cyclotomic exponents are recovered on demand, and families additionally use recurrences between coefficients. Weighted coupling functions keep their $q$-power exponents beside the rule.}\label{fig:architecture}
\end{qrwidefigure}

\subsection{Prime and cyclotomic factorization}\label{sec:cyclo}
Prime factorization underlies exact classical evaluation in WIGXJPF and WignerSymbols.jl~\cite{wig}. Legendre's formula gives the exponent of a prime $p$ in $n!$ as $\sum_{r\ge1}\lfloor n/p^r\rfloor$. Products and quotients of factorials therefore become sums and differences of integer exponents, and common factors cancel before large numbers are formed. Cyclotomic factorization provides the corresponding representation for quantum factorials~\cite{dcr}.

Let $t=q^2$. The polynomial $t^n-1$ factors over the integers into \emph{cyclotomic polynomials},
\begin{equation}\label{eq:cyclopoly}
 t^n-1=\prod_{d\mid n}\Phi_d(t),\qq \Phi_d(t)=\prod_{\zeta}(t-\zeta),
\end{equation}
where $\zeta$ runs over the primitive $d$th roots of unity. Here, $\Phi_d$ denotes the $d$th cyclotomic polynomial. For example,  $\Phi_1=t-1$, $\Phi_2=t+1$, $\Phi_3=t^2+t+1$, $\Phi_4=t^2+1$, and so on. Each $\Phi_d$ has integer coefficients and is irreducible over $\Q$, so the $\Phi_d$ play the role of primes~\cite{washington}. Since $[n]=q^{1-n}(t^n-1)/(t-1)$, each $\Phi_d$ with $d\ge2$ divides $[m]$ exactly when $d\mid m$, and hence divides $[n]!$ once for each of the $\lfloor n/d\rfloor$ multiples of $d$ up to $n$~\cite{dcr}. Thus
\begin{equation}\label{eq:cyclofact}
 [n]!=q^{-n(n-1)/2}\prod_{d=2}^{n}\Phi_d(t)^{\lfloor n/d\rfloor}.
\end{equation}  
For a product $R=\prod_a[n_a]!^{c_a}$, the cyclotomic exponents follow directly from the pairs $(n_a,c_a)$ as
\begin{equation}\label{eq:cyclocoords}
\begin{aligned}
 R&=q^{P(R)}\prod_{d\ge2}\Phi_d(t)^{E_d(R)},\\
 E_d(R)&=\sum_a c_a\lfloor n_a/d\rfloor,
\end{aligned}
\end{equation}
where the power $P(R)$ of $q$ can be read from \eqref{eq:cyclofact} in the same way. At $t=1$, $\Phi_d(1)=p$ when $d=p^r$ is a power of a prime $p$ and $\Phi_d(1)=1$ otherwise, so the exponent of $p$ in $R(1)$ is $\sum_{r\ge1}E_{p^r}(R)$. Prime factorization is thus the classical limit of cyclotomic factorization.

For evaluation the factorial rule is the more compact representation. The exponent vector of a single $[n]!$ has $n-1$ nonzero entries, so the expanded factors of a sum grow with the labels and with the number of terms. The rule has the same number of factorials at every spin. Any exponent can still be recovered from the rule by \eqref{eq:cyclocoords} with a few integer divisions. QRecoupling therefore evaluates sums from the rule alone and does not store the deferred cyclotomic representation of Ref.~\cite{dcr}. The exponents are useful for addressing two questions about individual factors.
\begin{itemize}
\item \emph{Valuations, zeros and poles.} At level $k$, $t=e^{2\pi i/h}$ with $h=k+2$ is a root of $\Phi_h$ and of no other $\Phi_d$. Thus $E_h(R)$ is the order of vanishing of an individual product $R$, with negative values indicating poles. Read from the prefactor and summands, these valuations identify vanishing and singular factors before evaluation. They do not decide whether the surviving terms cancel, or whether poles cancel between terms (\cref{sec:polynomialsupport}).
\item \emph{Exact prefactors.} For expansions in $x$ (\cref{sec:exact}), the exponents of the prefactor and of the first term are combined, so that common factors cancel before any polynomial is formed, and the odd exponents determine the radical part of the square root. Only this single product is expanded into exponents. At $q=1$ the prime exponents play the same role.
\end{itemize}
These exponents describe cancellation within a term, whereas cancellation between different terms is addressed in \cref{sec:numerics}. The analogy with primes holds in the variable $t=q^2$, since $\Phi_d(q^2)$ can factor further in $q$.

\subsection{Polynomial terms and level support}\label{sec:polynomialsupport}
Each quantum integer is an integer polynomial in $x=q+q^{-1}$, so a product of quantum factorials with integer exponents is a ratio of two such polynomials. The cyclotomic exponents decide when the denominator cancels, which gives a criterion for exact polynomial arithmetic without expanding the product.
\begin{proposition}[Polynomial factorial products]\label{prop:polynomialcriterion}
Let $R(q)=\prod_a[n_a]!^{c_a}$, where $n_a\in\Z_{\ge0}$ and $c_a\in\Z$. Then $R(q)\in\Z[x]$ if and only if $E_d(R)\ge0$ for every $d\ge2$. Any finite integer linear combination of such polynomial products also belongs to $\Z[x]$.
\end{proposition}
\begin{proof}
For $d\ge2$ the polynomial $\Phi_d$ is palindromic of degree $\varphi(d)$, where $\varphi$ is Euler's totient. Thus, the balanced factor $\Psi_d(q)=q^{-\varphi(d)}\Phi_d(q^2)$ is invariant under $q\mapsto q^{-1}$. It therefore lies in the invariant ring
\begin{equation}\label{eq:invariant}
 \Z[q,q^{-1}]^{q\leftrightarrow q^{-1}}=\Z[x],
\end{equation}
because every symmetric integer Laurent polynomial is an integer linear combination of $1$ and $C_n=q^n+q^{-n}$, with $C_0=2$, $C_1=x$ and $C_{n+1}=xC_n-C_{n-1}$. Each quantum factorial is also invariant, so the powers of $q$ in \eqref{eq:cyclocoords} combine to give $R=\prod_{d\ge2}\Psi_d(x)^{E_d(R)}$. The $\Psi_d$ are monic in $x$ and pairwise coprime, since their roots in $q$ are distinct roots of unity. Hence $R$ is a polynomial exactly when no exponent is negative.
\end{proof}

The proposition is an identity for generic $q$, so each target substitutes a specific value of $x$. At a unitary level, $x=2\cos(\pi/h)$ and a polynomial product becomes an algebraic integer of the real cyclotomic field of \cref{sec:levelvalues}~\cite{dcr}. At $q=1$, $x=2$ and it becomes an integer. In that limit the condition is sufficient but no longer necessary, because only the sums $\sum_rE_{p^r}(R)$ have to be nonnegative for $R(1)$ to be an integer. For example, $[2]!^2[3]!/[4]!=x/(x^2-2)$ is not a polynomial, although its classical value is $1$.

At a fixed unitary level $k$, the single exponent $E_h(R)$, with $h=k+2$, classifies a product. The product $R$ vanishes when $E_h>0$, is finite and nonzero when $E_h=0$, and has a pole when $E_h<0$. When the prefactor is regular and no summand is singular, only the terms with $E_h=0$ need to be evaluated, although their sum may still vanish through cancellation. Sums whose terms have poles fall outside this classification, because poles can cancel between terms, and they would need a separate algebraic regularization.

The $6j$ rule is the standard case in which both polynomiality and regular level support hold~\cite{kr,kl,tv}. For a classically admissible labelling, the balance $\sum_i\alpha_i=\sum_r\beta_r$ makes each unnormalized term $T(z)$ of \eqref{eq:racah} equal to $[z+1]$ times a quantum multinomial coefficient, so $T(z)$ and the Racah sum $S=\sum_z(-1)^zT(z)$ belong to $\Z[x]$. If the labelling is also admissible at level $k$, every factorial argument of $\cal N$ and of the denominator of $T(z)$ is below $h$, so $\cal N$ is finite and nonzero, while $[z+1]!$ contains the factor $[h]=0$ once $z\ge k+1$. Those terms vanish, and the contributing range is
\begin{equation}\label{eq:support}
\max_i\alpha_i\le z\le\min(\min_r\beta_r,k).
\end{equation}
The normalization $\cal N$ still carries rational factors and radicals, so polynomiality holds for the sum $S$ and not for the full symbol. The package uses both facts. It truncates the Racah sum to the range \eqref{eq:support} before evaluation at a level, and the expansion of a value in $x$ relies on the polynomial sum $S$.

\section{Numerical evaluation and accuracy}\label{sec:numerics}

Numerical evaluation must control both the size of intermediate values and the loss of accuracy through cancellation. Scaling addresses the first problem, while adaptive evaluation addresses the second by starting with inexpensive arithmetic and increasing precision when the estimated error is too large. Each evaluation of the sum constitutes a \emph{pass}, whose estimated error determines whether further evaluation is needed. We first describe this procedure for unweighted rules at a principal unitary level before turning to coupling coefficients and other targets.

\subsection{Scaled representation}\label{sec:scaling}
For $0\le n<h$, the level tables store the nonvanishing factorials in the form
\begin{equation}\label{eq:split}
[n]!=m_n\,2^{e_n},\qq \tfrac12\le m_n<1,\q e_n\in\Z.
\end{equation}
Both the plain and compensated passes use this split representation. The plain pass retains one floating-point mantissa, $m_n^{\rm hi}$. The compensated pass also retains a correction $m_n^{\rm lo}$, giving $(m_n^{\rm hi}+m_n^{\rm lo})2^{e_n}$. This pair of a leading value and a small correction is called a \emph{double word}. Small tables and wider output types are constructed with 64 guard bits. Larger machine-precision level tables use double-word arithmetic. Inverse factorials have the same representation, so products combine mantissa multiplications with integer exponent additions. On a contributing segment, the sum takes the form

\begin{equation}\label{eq:scaledsum}
S_{\rm seg}=a\,2^{e}\sum_{j=0}^{N}\tau_j,\qq
\tau_0=1,\q \tau_{j+1}=\tau_j\,r_{z_0+j},
\end{equation}
where $a\,2^e={\cal N}\,t_{z_0}$ includes the sign of the first term. The ratios use tables of $[n]$ and $1/[n]$. During forward summation, a term exceeding $2^{256}$ triggers a rescaling of the running term and partial sum by $2^{-256}$, with a corresponding increase in $e$. Such binary rescaling is exact within the normal floating-point range. It controls intermediate growth but does not recover digits lost to cancellation or extend the range of the final output type. Each pass requires $N$ ratio steps, with eight quantum-integer factors per step for the $6j$ symbol.

For small labels, an exponent bound can establish that all intermediate products remain safely within range. A \emph{short pass} then evaluates the same operations on unsplit factorials, avoiding the separate exponent bookkeeping. Otherwise, the scaled plain pass is used.

\subsection{Cancellation and conditioning}\label{sec:cancellation}
For a nonzero sum $S=\sum_z t_z$, define
\begin{equation}\label{eq:kappa}
\kappa=\frac{\sum_z|t_z|}{|S|}
\end{equation}
as its \emph{condition number}~\cite{dcr,higham}. Relative perturbations of at most $\epsilon$ in each term change $S$ by at most $\epsilon\kappa$ in relative terms, so $\kappa$ measures how strongly the sum amplifies errors in its terms. For unit roundoff $u=2^{-53}$, the characteristic error scales are $u\kappa$ for a plain pass and $u+u^2\kappa$ for a compensated pass~\cite{ors,comp}, to leading order in precision and cancellation. Full bounds also depend on the operation count and on the accuracy of the input terms.

For the computed value $\widehat{S}$ of a nonzero sum $S$, the relative error
\begin{equation}\label{eq:relativeerror}
\varepsilon_{\rm rel}=\frac{|\widehat{S}- S|}{|S|}
\end{equation}
measures the absolute error as a fraction of the sum. In \cref{fig:accuracy} the sum is the $6j$ symbol $\mathcal W$ of \eqref{eq:racah}, and its exact value is represented by an independently checked high-precision reference. The relative error is undefined at an exact zero, so zero decisions are tested separately.

The accuracy experiment uses 1209 nonzero inputs selected before evaluation. They comprise 245 equal-spin labels at $k=4j+1$, 120 random and 40 fusion-boundary samples at each of the levels $k=20,50,100,200,500,1000$, and 4 small nonzero values. Each reference is a direct factorial Racah sum, independent of package tables, ratios and recurrences, evaluated at 2048 and 3072 bits and required to agree to 200 decimal digits. Compensation delays the loss of accuracy but eventually loses digits too, whereas the adaptive evaluator of \cref{alg:numeric}, which repeats a rejected sum in more accurate arithmetic, keeps the largest relative error at $\varepsilon_{\rm rel}=3.4\times10^{-14}$ and preserves nonzero values down to $1.4\times10^{-204}$. A separate algebraic-number evaluation confirms 5 exact zeros and 10 nonzero controls at adjacent levels, and both the numerical output and the zero query classify all of them correctly.

\begin{qrwidefigure}
\centering
\includegraphics[width=\textwidth]{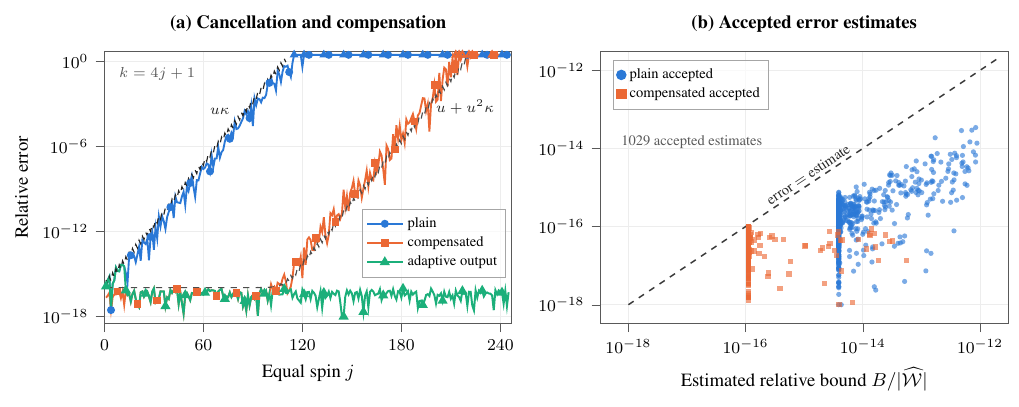}
\caption{Accuracy of the level-$k$ $6j$ evaluator. (a) Relative error of the plain pass, the compensated pass and the adaptive output for equal spins $j$ at $k=4j+1$, with the guides $u\kappa$ and $u+u^2\kappa$ as thin dark lines. (b) Relative error against the estimate $B/|\widehat{\mathcal W}|$ for the results accepted from a plain or compensated pass. Points below the diagonal have an error smaller than their estimate. Errors below $10^{-18}$ are drawn at that floor, and triangles in (a) mark errors above the plotting limit.}
\label{fig:accuracy}
\end{qrwidefigure}

\subsection{Adaptive evaluation and acceptance tests}\label{sec:escalation}
Each pass produces a candidate value $v$ for the full coefficient $\mathcal W$ and an estimate $B$ of its absolute error. The estimate is accumulated during the pass from the magnitudes of the terms and partial sums. It accounts for the rounding errors of forming the terms, summing them and applying the normalization. \Cref{app:error} derives the error estimates for the plain pass.

A candidate is \emph{accepted} at a threshold $\eta$ when $v$ and $B$ are finite, $v\ne0$ and $B\le\eta|v|$. It is then returned as $\widehat{\mathcal W}=v$. A computed zero never passes this relative test, so zeros are decided by exact means (\cref{sec:zerodecisions}).

\Cref{alg:numeric} gives the order of the evaluation stages. Level valuations first select the contributing terms and detect singular ones. The short or plain pass is accepted at $\eta=2^{-40}$. If it fails, the compensated pass uses double-word ratios and propagates rounding corrections through the Horner nesting~\cite{ors,comp,dekker,joldes}, with acceptance threshold $\eta=2^{-44}$. Cheap structural tests are interleaved with these stages to detect exact cancellation.

When neither pass succeeds, a $6j$ recurrence offers an alternative to the factorial sum. If that route is unavailable or declines, exact tests decide whether the sum vanishes (\cref{sec:zerodecisions}). A nonzero sum is then evaluated in three- and four-word arithmetic~\cite{hida}, followed when necessary by \texttt{BigFloat} arithmetic~\cite{mpfr} at doubling precision. Each stage is reached only when the preceding stages have not returned a result.

\begin{algorithm}[h]
\caption{Adaptive evaluation of $\mathcal W$ at level $k$}\label{alg:numeric}
\begin{algorithmic}[1]
\Require unweighted factorial rule $\mathcal R$, level $k$
\Ensure approximation $\widehat{\mathcal W}$, or an exact zero
\State keep the terms with $E_h=0$ \Comment{error on a pole}
\IfThen{no term remains}{\Return $0$}
\State $(v,B)\gets$ short pass if it applies, else plain pass
\IfThen{$v$ is accepted at $2^{-40}$}{\Return $v$}
\IfThen{the terms cancel in pairs}{\Return $0$}
\State $(v,B)\gets$ compensated pass in double words
\IfThen{$v$ is accepted at $2^{-44}$}{\Return $v$}
\IfThen{level reflection proves cancellation}{\Return $0$}
\IfThen{a $6j$ recurrence is accepted}{\Return its $v$}
\If{a modular image does not exclude zero}
  \IfThen{exact arithmetic gives zero}{\Return $0$}
\EndIf
\For{3 and 4 words, then $p,2p,4p,\ldots$ bits}
  \State $(v,B)\gets$ plain pass in that arithmetic
  \IfThen{$v$ is accepted}{\Return $v$}
\EndFor
\end{algorithmic}
\end{algorithm}

The compensated pass has the stricter threshold because its estimate is less conservative than that of the plain pass. \Cref{app:error} derives the plain-pass estimate under stated arithmetic assumptions, but not the compensated one, and neither is a certified enclosure. For all 1029 accepted results in \cref{fig:accuracy}(b) the error stayed below its estimate and reached at most 0.895 of it. The other 180 results were returned by later stages and are not plotted.

\subsection{Recurrence relations}\label{sec:recurrences}

The condition number $\kappa$ describes the factorial sum, but the same coefficient may admit a better-conditioned representation. While compensation and wider arithmetic improve the evaluation of that sum, a recurrence instead relates the coefficient to neighbouring values. With five labels fixed, the $6j$ symbols $f(x)=\sixj{x}{j_1}{j_2}{j}{j_3}{j_4}_q$ form a \emph{column} over the admissible values of $x$, satisfying the quantum analogue of the Schulten--Gordon recurrence~\cite{sg,tw}
\begin{equation}\label{eq:recurrence}
 A_x f(x+1)+B_xf(x)+A_{x-1}f(x-1)=0,
\end{equation}
where $A_x^2$ and $B_x$ are rational functions of quantum integers. Evaluating the column by stable propagation~\cite{gautschi} avoids the cancellation in the Racah sum, although its accuracy still depends on the propagation length and the size of the requested entry relative to the column. Classically, such relations also support asymptotic expansions~\cite{dupuislivine}, and a related recurrence is known for the squared symbol~\cite{bonzomlivine}.

For an individual $6j$ symbol, the evaluator selects a short path from a column end and rejects the result if the requested entry is too small relative to the column. Although the shortest path has as many steps as the Racah sum has ratios, the cost and conditioning of those steps differ. For a complete fusion matrix, recurrence becomes the primary route: columns are propagated from both ends, matched and normalized by orthogonality, so that entries share computational work (\cref{sec:performance}). Coupling coefficients admit analogous recurrences in a magnetic label or the coupled spin~\cite{kr}, but a general factorial rule need not possess such an alternative representation.

\subsection{Exact zero decisions within numerical evaluation}\label{sec:zerodecisions} After selection rules and level valuations remove empty or termwise-vanishing cases, write $\mathcal W={\cal N} S$ with a regular, nonzero prefactor $\cal N$. Then $\mathcal W=0$ exactly when the finite sum $S$ vanishes, so the square-root normalization need not be constructed. As in adaptive geometric predicates~\cite{shewchuk}, exact arithmetic is used only when floating-point evaluation with an error estimate cannot decide.

To decide whether the sum vanishes, the evaluator first looks for structural identities. For instance, equal terms with opposite signs cancel in pairs, and a level-reflection identity can prove cancellation directly. If neither identity settles the question, the sum is evaluated in a finite field. Provided the denominators remain invertible, a nonzero residue proves $S\ne0$. A zero residue proves nothing, because a nonzero value can map to zero modulo the chosen prime, and the same holds when the modular image cannot be formed. In both cases the sum is returned as zero only after exact confirmation.

The confirming arithmetic follows the target:
\begin{itemize}
\setlength{\itemsep}{2pt}\setlength{\parsep}{0pt}
\item At $q=1$, integer Horner summation gives $S/t_{z_0}=P/Q$, with $Q\ne0$. Cancellation is exact precisely when $P=0$.
\item At level $k$, with $h=k+2$ and $\zeta=e^{i\pi/h}$, each quantum integer is $[n]=(\zeta^n-\zeta^{-n})/(\zeta-\zeta^{-1})$. Write the signed ratio \eqref{eq:ratio} between contributing terms as $A_z/B_z$, with $A_z,B_z$ integer combinations of powers of $\zeta$. Starting from $P=Q=1$, set
\begin{equation}\label{eq:clearedhorner}
P\gets B_zQ+A_zP,\qq Q\gets B_zQ,
\end{equation}
from the last term to the first, using the previous values of $P$ and $Q$ in both updates. The vector of $P$ has $h$ integer coefficients, since $\zeta^h=-1$. The sum vanishes exactly when the polynomial with these coefficients is divisible by the cyclotomic polynomial $\Phi_{2h}$. Machine integers are used when a coefficient bound guarantees that they suffice, and larger cases use arbitrary-precision integers. The phases $q^{wz}$ of \code{qcg} and \code{q3j} enter the same recursion as shifts of the coefficients.
\item For real or complex $q$, finite binary components are represented as exact rationals. Horner summation then decides zero at that stored parameter, including integer $q$-power weights. For a floating-point input, this is an exact decision about the represented value of $q$, whereas \code{Level(k)} retains the exact root of unity.
\end{itemize}
If exact arithmetic excludes zero, numerical evaluation continues to higher precision. Vanishing identically in $q$, which \code{iszero} decides for a \code{Symbolic()} form, is settled without expanding the sum. After clearing denominators and negative powers of $q$, the sum becomes an integer polynomial. A nonzero polynomial cannot vanish at a positive integer larger than the sum of the absolute values of its coefficients. A bound derived from the rule selects such an integer, and the polynomial is evaluated there modulo word-sized primes. One nonzero residue disproves the identity. Conversely, zero residues prove the identity once the product of the primes exceeds a proven bound on the absolute value of this integer evaluation.

\subsection{Coupling coefficients and other targets}\label{sec:complex}
A weighted coupling coefficient at a level is evaluated from \eqref{eq:weightedrule} by a shorter sequence of stages than \cref{alg:numeric}. Its quantum integers are stored as real high and low parts, and the phases $q^{wz}=e^{i\pi r/h}$, with $r\equiv wz\pmod{2h}$, are read from a table instead of being formed as powers of a rounded $q$. The first stage is already a compensated pass, which keeps each term, its phase and the running sum in double words and is accepted for $\kappa$ up to about $10^{15}$. For larger $\kappa$ the coefficient is taken from a recurrence in the magnetic label, and if that declines the sum is repeated in \texttt{BigFloat} arithmetic at doubling precision. A sum that still cancels there is tested for an exact zero (\cref{sec:zerodecisions}).

The classical target \code{Classical()} uses the same sequence of plain and compensated passes, with exact integer evaluation as its final fallback. At real or complex $q$, the evaluator uses running bounds and wider arithmetic without level classification. Real \texttt{Float64} tables are built in double words and rounded once per entry. Every kernel evaluates the parameter that was supplied, including its deviation from one. An input intended at high precision must therefore be constructed at that precision, because converting a rounded input does not restore missing digits.

Complex deformation also requires a branch convention. Each quantum integer $[n]$ is the product of the balanced factors $\Psi_d$ (\cref{sec:polynomialsupport}) over the divisors $d>1$ of $n$, and the prefactor takes the principal square root of each $\Psi_d$ separately. Replacing these by the principal square root of their product can change the sign, so comparisons must use the same branches.

\section{Exact and symbolic computation}\label{sec:exact}

Exact computation serves three practical purposes: distinguishing algebraic zeros from small values, verifying identities without a tolerance, and retaining a formula for later specialization. Because these tasks do not always require an expanded expression, the package separates the compact representation from the operations used to evaluate, compare or transform it.

The target \code{Exact()} returns the classical value, \code{Exact(k)} retains the value at a specified level, and \code{Symbolic()} retains its dependence on $q$. Their equality and zero tests use integer and modular arithmetic without a computer-algebra system. When an expanded representation is needed, the optional Nemo backend~\cite{nemo} provides expansions in $x=q+q^{-1}$, nested-radical expressions when available, and the algebraic operations used by \code{prove_identity}. The following subsections explain the representations and the arithmetic they support.

\subsection{Classical values}
At $q=1$ the ratios are rational numbers, and a backward Horner recursion accumulates the sum as a numerator and denominator in integer arithmetic, with one reduction at the end. The prefactor and the first term are reduced through prime valuations (\cref{sec:cyclo}), so common factors cancel before large products are formed. The result is a sign and a reduced rational number under a square root, $v=\sigma\sqrt r$. Sums of such values with unlike roots stay exact as sums of square roots with rational coefficients, so \code{==} and \code{iszero} remain exact after arithmetic. The listing below illustrates the classical exact target:
\begin{lstlisting}
a = q6j(Exact(), 1, 1, 1, 1, 1, 1)    # == 1//6
b = q3j(Exact(), 2, 2, 0, 1, -1, 0)   # == -√(1//5)
cg0 = qcg(Exact(), 1//2, 1//2, 1//2, -1//2, 0, 0)
@assert b^2 == 1//5
@assert cg0^2 == 1//2
\end{lstlisting}

\subsection{Values at a level}\label{sec:levelvalues}

At $q=e^{i\pi/h}$ the quantum integers satisfy $[h-n]=[n]$ and $[n+h]=-[n]$. Every nonvanishing quantum integer therefore reduces, up to a sign, to a label $1\le n\le\lfloor h/2\rfloor$, where $[n]>0$. An \code{Exact(k)} value holds the factorial rule and the level, and it is written in these quantum integers only when it is shown. Factors that vanish at the level cancel between the numerator and the denominator of a contributing term and leave their limit $[jh]/[h]=(-1)^{j-1}j$. For each prime $p$ dividing $h$, the distribution relation
\begin{equation}\label{eq:distribution}
 [pb]\prod_{j=1}^{p-1}\Bigl[\frac{jh}{p}\Bigr]=p\prod_{j=0}^{p-1}\Bigl[b+\frac{jh}{p}\Bigr],
\end{equation}
which follows from $\sin(p\theta)=2^{p-1}\prod_{j=0}^{p-1}\sin(\theta+j\pi/p)$, writes larger labels through smaller ones. These identities never change the value. Two equal values can still be written differently, so equality is decided by computation and not by reading the form.

Quantum integers are symmetric Laurent polynomials in $q$, so by the identity \eqref{eq:invariant} they are polynomials in $x=q+q^{-1}$ with integer coefficients. In particular, $[n]=U_{n-1}(x/2)$, where $U_{n-1}$ is a Chebyshev polynomial of the second kind. Pairing inverse roots of unity gives the real cyclotomic factors $\psi_e(x)$, $e\ge3$, the minimal polynomials of $2\cos(2\pi/e)$. Their exponents in a factorial product are $E_{e/\gcd(e,2)}$ from \eqref{eq:cyclocoords}. At level $k$ these quantum integers lie in
\begin{equation}\label{eq:field}
 x_h=2\cos\left(\frac{\pi}{h}\right),\q K_h=\Q(x_h)\simeq\Q[x]/(\psi_{2h}),
\end{equation}
the maximal real subfield of $\Q(\zeta_{2h})$, of degree $d=\varphi(2h)/2$, where $\varphi$ is Euler's totient~\cite{washington}. Since $[n]$ is monic of degree $n-1$ in $x$, the quantum integers $[1],\ldots,[d]$ form a basis of $K_h$. A normalized unweighted coefficient can require a square root over this field. Its square lies in $K_h$, but the coefficient itself need not.

For a given set of spins \code{js}, its exact value, its expansion in $x$, its nested-radical form and its numerical approximation are separate requests:

\begin{lstlisting}
using Nemo                     # optional backend
js = (1,1,1,1,1,1)
v = q6j(Exact(10), js...)      # ([5] − 1)/[4]!
p = xform(v)                   # (2x² − 7)/3
r = radical(v)                 # (2√3 − 3)/3
@assert numeric_value(v) ≈ q6j(Level(10), js...)
\end{lstlisting}

For an individual unweighted coefficient, the expansion writes the value as
\begin{equation}\label{eq:exactform}
 v=P(x_h)\sqrt{R(x_h)},\qq P,R\in\Q[x],
\end{equation}
with both polynomials reduced below degree $d$, where the odd $\psi_e$ exponents identify the factors retained under the square root. 
A nested-radical form does not always exist. Since $K_h/\Q$ is abelian, a classical Galois-theoretic argument gives the criterion~\cite{washington,wantzel}.
\begin{proposition}\label{prop:radicals}
Let $v$ be a real algebraic number with $v^2\in K_h$. Then $v$ is expressible by rational numbers and nested real square roots if and only if $[\Q(v^2):\Q]$ is a power of two.
\end{proposition}
\begin{proof}
Every subfield $L=\Q(v^2)$ of the abelian field $K_h$ is Galois over $\Q$ with abelian group. If $[L:\Q]$ is a power of two, a chain of subgroups realizes $L$ as a tower of quadratic extensions, and the real square root of $v^2$ adds at most one step. Conversely, if $v$ lies in a tower of quadratic extensions, then $\Q(v^2)$ is a subfield of that tower and its degree divides a power of two.
\end{proof}
The function \code{radical} returns \code{NoRadical} when the criterion fails or the descent exceeds its budget, and the result records which of the two occurred.

Unweighted exact level values support addition, subtraction and multiplication. Their field coefficients are expressed in the basis $[1],\ldots,[d]$, with square roots retained separately. Addition can produce several radical components. Division is available when the divisor has a single nonzero radical component. Inverting a sum with unlike roots is not yet implemented. Field coordinates are recovered from evaluations at conjugates modulo word-sized primes $p\equiv1\pmod{2h}$. Bounds on the conjugates of the cleared expression establish when the reconstruction is exact.

For individual real forms, equality can be decided directly from the factorial rules without constructing field coordinates. An unequal residue of their squares proves inequality. Agreement modulo sufficiently many primes establishes equality of the squares, after which a certified sign test distinguishes the two possible values. Sums with several radical components use separate exact tests. Square roots and half-integer powers are supported when the result can be retained as a single-product form. Arbitrary nested roots are outside this representation.

\begin{remark}[Exact arithmetic and identities]\label{rem:exactarithmetic}

Fusion rules, dimension identities and coherence relations of $U_q(\mathfrak{sl}_2)$, such as orthogonality and the Biedenharn--Elliott (pentagon) identity~\cite{kr,kl}, can therefore be verified exactly. \Cref{sec:coherenceworkflow} gives examples. The forms of \code{qcg} and \code{q3j} carry powers of $q$. Equality and zero tests cover them, although they have no arithmetic yet. Exact decisions also distinguish a vanishing recoupling coefficient from a small numerical value. The query \code{iszero_at}, and \code{iszero} on an \code{Exact(k)} value, settle this by the tests of \cref{sec:zerodecisions} without forming the value.

\end{remark}

\subsection{Generic expressions}

\code{Symbolic()} returns the formula as a function of $q$, written in quantum integers without carrying out the sum. For generic $q$ the quantum integers $[n]$ with $n\ge2$ are multiplicatively independent, because each is the product of the balanced factors $\Psi_d$ over the divisors $d>1$ of $n$ (\cref{sec:polynomialsupport}), and these factors are coprime. A product of quantum factorials therefore has a single set of exponents in the $[n]$, read from the rule by counting. A form \code{s} is evaluated at any $q$ as \code{s(q)}, and \code{s(1)} is the exact classical value. The operations \code{==} and \code{iszero} test identities wherever the expressions are regular and their branches agree, using the method of \cref{sec:zerodecisions}. Supported single-product forms remain compact under multiplication, division and integer powers. Square roots and half-integer powers are retained when they fit this representation. Adding general symbolic forms requires either a specified value of $q$ or an explicit algebraic expansion.  With Nemo loaded, \code{xform} expands the form in $x$.

\begin{lstlisting}
s = q6j(Symbolic(), js...)   # ([5] − 1)/[4]!
@assert s(1) == 1//6
f = xform(s)    # (x² − 3)/(x⁴ − 3x² + 2)
\end{lstlisting}

At $x=2$, that is $q=1$, the expansion gives $1/6$, and at $x=2\cos(\pi/12)$ it reduces to the level value above. A generic identity proved in $\Q(x)$ therefore holds at every level at which its denominators are nonzero and its branches compatible, and it does not have to be proved again level by level. The generic identity in \cref{sec:coherenceworkflow} illustrates this use.

\section{Computational workflows}\label{sec:workflow}

The representations developed above serve different computational needs. Recurrences allow families of coefficients to share computational work, while exact values support identity checks without tolerances. Recoupling matrices connect these coefficients to applications by transporting observables between coupling bases and supplying the transformations used in tensor calculations. The examples below illustrate these uses, then show how the same interface extends to a user-defined finite sum.

\subsection{Batches, families and precision}
A vector of labels can be evaluated in one call, which shares the tables of its target and gives each worker thread its own scratch storage. For sequential scalar calls, the same reuse is available through a workspace (\code{EvaluationWorkspace}). In the listing below the target also requests wider precision, and \code{fmatrix} computes a whole family through the recurrence \eqref{eq:recurrence}.
\begin{lstlisting}
using QRecoupling, LinearAlgebra
L = collect(all_6j(k = 10))   # 17,732 labels
w = q6j(L; k = 10)            # one batch
wb = q6j(Level(10; T = BigFloat), js...)
F, e, f = fmatrix(40, 35, 30, 45; k = 300)
@assert w[5] == q6j(Level(10), L[5]...)
@assert transpose(F) * F ≈ I
\end{lstlisting}
The first assertion compares a batch entry with its scalar evaluation. The second checks orthogonality of the $61\times61$ matrix, whose residual is about $10^{-15}$.

\subsection{Fusion data and exact identities}\label{sec:coherenceworkflow}
At level $k=3$, the integer-spin sector contains the Fibonacci fusion rule $1\otimes1=0\oplus1$~\cite{nayak}. Writing $\varphi=(1+\sqrt5)/2$, the package convention gives
\begin{equation}\label{eq:fib}
 d_1=\varphi,\qq
 F^{111}_1=\begin{pmatrix}
 \varphi^{-1}&-\varphi^{-1/2}\\
 -\varphi^{-1/2}&-\varphi^{-1}
 \end{pmatrix}.
\end{equation}
Signs of individual entries depend on the fusion-basis gauge, fixed here by \eqref{eq:fmatrix}, so comparison with another anyon convention must account for it~\cite{kitaev,nayak}.
\begin{lstlisting}
F, rows, cols = fmatrix(1, 1, 1, 1; k=3)
φ = (1 + √5)/2
Fref = [1/φ -1/√φ; -1/√φ -1/φ]
@assert rows == cols == [0, 1]
@assert F ≈ Fref
@assert transpose(F)*F ≈ I
d = qdim(Exact(3), 1)
@assert iszero(d^2 - d - 1)
\end{lstlisting}
The listing checks the channel ordering, the convention of the numerical matrix and the exact value of the dimension. The optional TensorKit extension uses these same fusion data, allowing tensor calculations to inherit the conventions checked here.

Recoupling identities~\cite{kr,kl} can be checked before any level is chosen. With $d_j=[2j+1]$, the orthogonality relation $\sum_e d_e\sixj{1}{1}{e}{1}{1}{1}_q^2=d_1^{-1}$ becomes a rational-function identity in $x=q+q^{-1}$, which can be established for generic $q$. Symbolic forms can be multiplied when each is a single product of quantum integers, but a sum needs a value of $q$, so the listing expands them with \code{xform}, using the Nemo backend loaded above.

\begin{lstlisting}
ds(j) = xform(qdim(Symbolic(), j))
ws(e) = xform(q6j(Symbolic(), 1, 1, e, 1, 1, 1))
lhs = sum(ds(e) * ws(e)^2 for e in 0:2)
@assert lhs == ds(1)^-1 
\end{lstlisting}

The assertion compares two rational functions of $x$, so it proves the orthogonality relation for these labels at generic $q$. 

\subsection{An anyonic chain with TensorKit}\label{sec:tensornetworks}

The TensorKit extension~\cite{tensorkit} turns recoupling data into tensor symmetries. The sector type \code{SU2k(k)} supplies the fusion rules of $\mathrm{SU}(2)_k$, with $F$-symbols from \code{fsymbol} and $R$-symbols from \code{rmatrix}. TensorKit uses these data to fuse, recouple and braid tensors. For positive real deformation, \code{SU2q(q)} supports all nonnegative integer and half-integer spins and provides real $F$-symbols for planar tensor operations. Unitary braiding is not provided for this sector.

An open Fibonacci chain gives a simple application~\cite{goldenchain}. At level $k=3$, the spin-$1$ object $\tau$ obeys $\tau\otimes\tau=0\oplus\tau$. Let $\Pi^{(0)}_{i,i+1}$ project adjacent sites onto the trivial fusion channel. We choose the Hamiltonian
\begin{equation}\label{eq:fibonacci-chain}
 H_L=\sum_{i=1}^{L-1}\Pi^{(0)}_{i,i+1},
\end{equation}
with positive unit coupling. The projectors act on different pairs, so adding them requires a common fusion basis. For any chain length $L\ge2$, the listing places the two-site projector on each neighbouring pair and adds identity operators on the remaining sites. TensorKit performs the required basis transformations using the supplied $F$-symbols.

\begin{lstlisting}
using QRecoupling, TensorKit, LinearAlgebra
function spectrum(τ, L)
    @assert L >= 2
    V = Vect[typeof(τ)](τ => 1)
    P0 = zeros(ComplexF64, V ⊗ V ← V ⊗ V)
    block(P0, one(τ)) .= 1  # trivial projector
    Id(n) = id(ComplexF64, V^n)  # n sites
    H = sum(1:L-1) do i
        term = P0
        i > 1   && (term = Id(i-1) ⊗ term)
        i < L-1 && (term = term ⊗ Id(L-i-1))
        term
    end
    E = [eigvals(Hermitian(Matrix(b)))
         for (_, b) in blocks(H)]
    return sort(reduce(vcat, E))
end
τ = SU2k(3)(1)  # physical spin, not doubled spin
@assert spectrum(τ, 5) ≈
        spectrum(FibonacciAnyon(:τ), 5)
\end{lstlisting}

The assertion compares the eigenvalues collected from all total-charge blocks with those obtained using TensorKit's built-in Fibonacci sector. These spectra are unchanged by a change of fusion-basis gauge, so agreement checks the assembled operator across the two implementations. The listing compares five-site chains, and the test suite also checks three sites. Replacing the sector by \code{SU2k(k)} constructs the corresponding spin-$1$ projector chain at any level $k\ge2$. The Fibonacci comparison applies only at $k=3$.

\subsection{An intermediate-channel observable}
A recoupling matrix also transports observables between coupling bases. At level $k$, the quantum Casimir of the intermediate channel $f$, after a shift and rescaling, has eigenvalues $\lambda_f=2\cos((2f+1)\pi/h)$, where $h=k+2$. In the other coupling basis, the operator $K=F\operatorname{diag}(\lambda_f)F^{\mathsf T}$ is tridiagonal, which is the recurrence \eqref{eq:recurrence} written as an eigenvalue equation~\cite{kr,tw,askeywilson,terwilliger}. This connects coefficient evaluation to an observable with a sparse representation. The following example checks the tridiagonal structure numerically.
\begin{lstlisting}
F, e, f = fmatrix(2, 3, 2, 3; k=20)
θ = π/22
λ = 2 .* cos.((2 .* f .+ 1) .* θ)
K = F * Diagonal(λ) * transpose(F)
offband = [K[i,j] for i in axes(K,1)
                    for j in axes(K,2)
                    if abs(i-j) > 1]
@assert maximum(abs, offband) < 1e-12
@assert transpose(F)*F ≈ I
\end{lstlisting}

\subsection{A user-defined terminating series}\label{sec:userdefined}

A finite sum of factorial products can be constructed and evaluated within the package. We illustrate this with the ${}_2F_1$ series \eqref{eq:hypergeometricexample}. For $n\ge0$ and integers $a,b\ge1$, its summand is
\begin{equation}\label{eq:hypergeometricfactorials}
\begin{aligned}
 \frac{(-n)_z(a)_z}{(b)_z\,z!}
 &=(-1)^z\frac{n!(b-1)!}{(a-1)!}\\
 &\q\times\frac{(a+z-1)!}{(n-z)!\,(b+z-1)!\,z!}.
\end{aligned}
\end{equation}
Replacing each factorial by its quantum counterpart gives the input below for its factorial rule. A \code{prefactor} pair \code{n => c} denotes $[n]!^c$, and a \code{factors} triple \code{(u,v,c)} denotes $[uz+v]!^c$. For example, \code{(-1,n,-1)} supplies $1/[n-z]!$.

\begin{lstlisting}
n, a, b = 3, 2, 4
rule = FactorialSum(0:n;
    prefactor=[n=>1, b-1=>1, a-1=>-1],
    factors=[(1,a-1,1), (-1,n,-1),
             (1,b-1,-1), (1,0,-1)],
    alternating=true)
@assert qeval(Exact(), rule) == 1//5
real_value = qeval(At(0.8), rule)
level_value = qeval(Exact(10), rule)
@assert numeric_value(level_value) ≈
        qeval(Level(10), rule)
expr = qeval(Symbolic(), rule)
@assert expr(1) == 1//5
\end{lstlisting}

The classical request returns ${}_2F_1(-3,2;4;1)=1/5$. At level $10$ the value is $(3-\sqrt3)/6$, with all factorial arguments below $h=12$, so the specialization is regular. The real and symbolic calls use the same factorial deformation defined in \cref{sec:architecture}, and the symbolic form evaluated at $q=1$ returns the classical value. Changing the rule selects another sum within \eqref{eq:generalrule} and requires no new evaluation kernel.

\FloatBarrier
\section{Performance}\label{sec:performance}
\begin{qrwidefigure}
\centering
\includegraphics[width=0.49\textwidth]{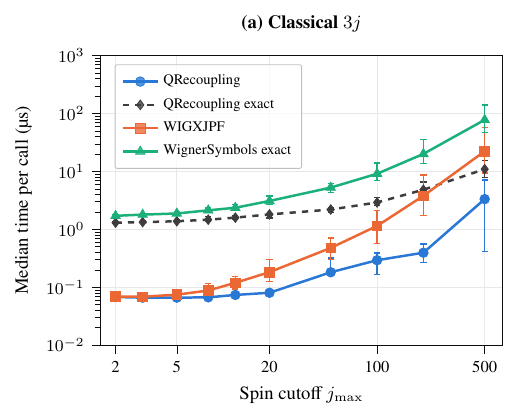}\hfill
\includegraphics[width=0.49\textwidth]{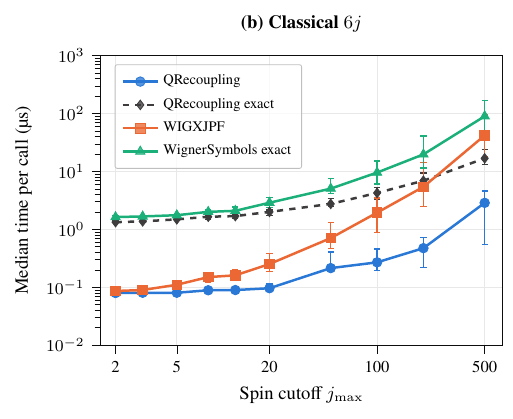}
\caption{Classical scalar cost. At each cutoff, 80 distinct ordered admissible inputs are sampled with largest spin in $(j_{\max}/2,j_{\max}]$. Exact zeros are retained separately and are not replaced, leaving 797 nonzero $3j$ and 798 nonzero $6j$ inputs across all cutoffs. Exact routes include \texttt{Float64} conversion.}
\label{fig:classicalcost}
\end{qrwidefigure}

This section reports benchmarks of QRecoupling.jl. They measure the cost of single coefficients, of matrices and batches, of exact output, and of setup. All measurements use QRecoupling~0.6.4, Julia 1.12.4 and an Apple M2 Pro, on one thread unless stated otherwise.
Scalar timings are medians of 41 samples per input, where a sample is the mean of 50 consecutive warmed calls, because a single call at small spins is too short to time reliably. Labels are prepared outside the timed call, and compilation and first table construction are measured separately. Figures show the median and interquartile range across inputs. Timed scalar coefficients are checked against independent high-precision references with relative-error target $2^{-40}$, the plain-pass threshold in \cref{sec:escalation}. Matrix, batch and exact workloads use the checks described below.

\subsection{Scalar coefficients}
\begin{qrwidefigure}
\centering
\includegraphics[width=0.49\textwidth]{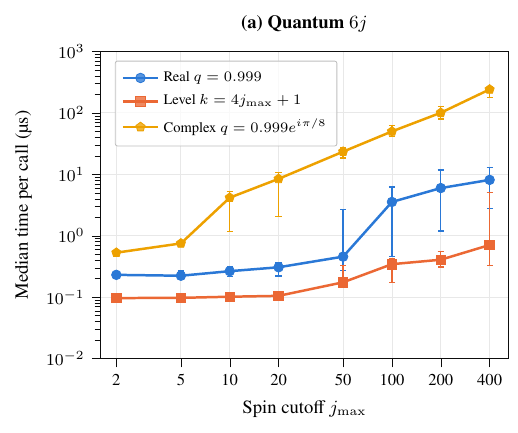}\hfill
\includegraphics[width=0.49\textwidth]{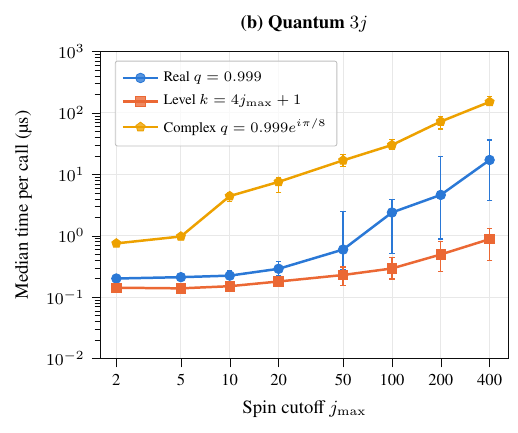}
\caption{Costs across numerical deformations: (a) $6j$ symbols through $j_{\max}=400$; (b) quantum $3j$ symbols through $j_{\max}=400$. Each experiment samples 80 distinct ordered admissible inputs per cutoff, shared across its three deformations, with largest spin in $(j_{\max}/2,j_{\max}]$ and $3j$ magnetic labels summing to zero. All plotted outputs meet the error target.}\label{fig:targets}
\end{qrwidefigure}

\Cref{fig:classicalcost} compares classical \code{q3j} and \code{q6j} with WIGXJPF~1.14~\cite{wig} and WignerSymbols.jl~2.0.0, all returning \texttt{Float64}. At the smallest spins QRecoupling and WIGXJPF are comparable, with medians within ten per cent of each other at $j_{\max}=2$. The ratio of medians then grows in favour of QRecoupling, to about $2.3$ for $3j$ and $2.6$ for $6j$ at $j_{\max}=20$, and to 6.7 and 15 at $j_{\max}=500$ (\cref{tab:classicalcost}). Up to $j_{\max}=20$ nearly all sums take the short pass (\cref{sec:scaling}), which explains the step in the QRecoupling curves between $20$ and $50$. The result cache of WignerSymbols is cleared before every call, so its samples are single calls, and the arithmetic tables of all routes stay warm. The plotted WignerSymbols route constructs the exact result before conversion. Its separate direct floating-output route fails the error target on 13 inputs at cutoff 500, including 6 returned zeros for nonzero references, and is not plotted.
\begin{table}[!htbp]
\centering\small
\caption{Median time per call (\textmu s) at selected spin cutoffs from \cref{fig:classicalcost}. QR denotes QRecoupling, WIG denotes WIGXJPF and WS denotes WignerSymbols. ``Exact'' includes construction and \texttt{Float64} conversion.}
\label{tab:classicalcost}

\begingroup\setlength{\tabcolsep}{3.5pt}
\begin{tabular}{@{}l r r r r r@{}}\toprule
Symbol & $j_{\max}$ & QR & QR exact & WIG & WS exact\\\midrule
$3j$ & 5 & 0.0658 & 1.38 & 0.0742 & 1.88\\
$3j$ & 50 & 0.181 & 2.19 & 0.477 & 5.25\\
$3j$ & 500 & 3.32 & 10.9 & 22.3 & 77.8\\
\midrule
$6j$ & 5 & 0.08 & 1.48 & 0.11 & 1.75\\
$6j$ & 50 & 0.213 & 2.75 & 0.704 & 5.04\\
$6j$ & 500 & 2.86 & 16.8 & 41.9 & 89.8\\
\bottomrule\end{tabular}\endgroup

\end{table}

The references are exact rational factorial sums times square roots of rational numbers, evaluated at 512 and 768 bits and agreeing to 100 decimal digits. Against them the largest relative errors of the default route are $3.6\times10^{-14}$ for $3j$ and $4.8\times10^{-14}$ for $6j$.

\paragraph{Quantum deformation.}
The deformed $6j$ experiment extends to $j_{\max}=400$ (\cref{fig:targets}), evaluating the same 80 random admissible inputs per cutoff at real $q=0.999$, at the principal level $k=4 j_{\max}+1$, and at complex $q=0.999 e^{i\pi/8}$. These values of $|q|$ keep the sampled coefficients within floating-point range. Together with the further real deformations of \cref{tab:quantumcost}, all 2720 outputs are nonzero and meet the error target, with largest error $3.5\times10^{-14}$. The references are direct Racah sums at 1024 and 1536 bits agreeing to 80 decimal digits, with complex prefactors on the branch of \cref{sec:complex}.
At the largest cutoff the level median is $0.707$~\textmu s and the complex median $241$~\textmu s.
The larger complex cost includes branch handling and adaptive arithmetic. These timings do not separate their contributions, and the relative cost depends on the deformation and the sampled labels.

The companion \code{q3j} experiment uses the same deformations and reference precisions, with independent weighted factorial sums. All 1920 outputs through $j_{\max}=400$ meet the error target, with largest error $1.5\times10^{-14}$. The function \code{qcg} shares this kernel and differs from \code{q3j} by the factor in \eqref{eq:threejcg}, so this experiment measures their shared summation engine rather than separately benchmarking both interfaces. A spin-half coproduct identity independently checks the convention. At $j_{\max}=400$ about half of the level inputs have $\kappa>10^3$, and the level median is $0.882$~\textmu s.

\Cref{tab:quantumcost} compares with QWignerSymbols.jl~0.1.0 at positive real $q$. That package receives $q^2$ to match conventions, and each reference uses the parameter that its route stores. At $j_{\max}=20$ both packages meet the target on every input, and at $q=0.95$ QRecoupling's median is about 390 times faster. At $j_{\max}=50$ QWignerSymbols meets the target on only a fraction of the inputs, so no timing is reported for it there.

\begin{table}[!htbp]
\centering\small
\caption{Median time per call (\textmu s) for real deformation. QR denotes QRecoupling and QWS denotes QWignerSymbols. Pass counts use relative error at most $2^{-40}$, and a dash means that some outputs failed this check.}
\label{tab:quantumcost}

\begingroup\setlength{\tabcolsep}{5pt}
\begin{tabular}{@{}r r r r r@{}}\toprule
$q$ & $j_{\max}$ & QR & QWS & QWS passes\\\midrule
0.8 & 20 & 0.3 & 119 & 80/80\\
0.8 & 50 & 0.431 & \textemdash & 1/80\\
\addlinespace[2pt]
0.95 & 20 & 0.301 & 119 & 80/80\\
0.95 & 50 & 0.427 & \textemdash & 62/80\\
\addlinespace[2pt]
0.999 & 20 & 0.307 & 121 & 80/80\\
0.999 & 50 & 0.457 & \textemdash & 78/80\\
\bottomrule\end{tabular}\endgroup

\end{table}

\subsection{Matrices, batches and threads}
The function \code{fmatrix} computes each column of an $F$-matrix by the recurrence \eqref{eq:recurrence}. A column is propagated from both ends of its admissible range, matched where the two runs meet, and orthogonality fixes the normalization. Each step costs a few multiplications, whereas a scalar entry costs a sum, or a recurrence path of its own, whose length grows with the spins. The columns also share most of the recurrence coefficients, because only one factor depends on the column label. At $n=401$ channels an entry costs $0.0250$~\textmu s classically and $0.0255$~\textmu s at a level, about 98--124 times faster than a loop of scalar \code{fsymbol} calls, even when that loop reuses an \code{EvaluationWorkspace} (\cref{tab:reuse}). The recurrence therefore saves far more than reusing scratch storage, although small matrices have less work over which to spread the setup cost.

\begin{qrwidetable}
\centering\small
\caption{Left: median cost per entry (\textmu s) of complete equal-spin $F$-matrices with $n=2s+1$ channels. ``Scalar'' evaluates every \code{fsymbol} independently and ``Scalar + work'' reuses a workspace, and every route allocates its output matrix. Right: median ns per coefficient for batches of $10^{6}$ classical or level-$400$ symbols, with the level gain relative to the one-thread batch.}
\label{tab:reuse}
{
\begingroup\setlength{\tabcolsep}{5pt}\begin{tabular}[t]{@{}r r r r@{}}\toprule
$n$ & \code{fmatrix} & Scalar & Scalar + work\\\midrule
\multicolumn{4}{@{}l}{Classical}\\
21 & 0.0303 & 0.12 & 0.12 \\
101 & 0.0251 & 0.43 & 0.431 \\
401 & 0.025 & 2.45 & 2.46 \\
\addlinespace[3pt]
\multicolumn{4}{@{}l}{Level $k=4s+1$}\\
21 & 0.0341 & 0.132 & 0.131 \\
101 & 0.0274 & 0.367 & 0.36 \\
401 & 0.0255 & 3.12 & 3.11 \\
\bottomrule\end{tabular}\endgroup
}\qq
{
\begin{tabular}[t]{@{}r r r r@{}}\toprule
Threads & Classical & Level & Level gain\\\midrule
1 & 398.1 & 320.8 & 1.00\\
2 & 213.0 & 172.6 & 1.86\\
4 & 117.4 & 97.0 & 3.31\\
6 & 85.0 & 71.8 & 4.47\\
\bottomrule\end{tabular}
}

\end{qrwidetable}

The comparison covers classical and principal-level targets with outer spins $(s,s,s,s)$ and $(s,s+\tfrac12,s+1,s+\tfrac32)$ through $s=200$, with the level chosen as one plus the sum of the outer spins: $k=4s+1$ for equal spins and $k=4s+4$ for unequal spins. Both shapes have $n=2s+1$ channels. Each route is timed by the median of 9 complete calls, including allocation of the output. All entries agree with both scalar loops, and 708 entries across 32 matrices, including edges and centres, were checked against independent factorial references evaluated at 1024 and 1536 bits. The largest entry error is $5.5\times10^{-17}$, and the largest orthogonality defect $\max|F^{\mathsf T}F-I|$ is $1.1\times10^{-15}$.

For independent coefficients, a batch shares tables and runs on threads. With $10^{6}$ admissible inputs at $k=400$ and spins at most 100, the cost falls from $320.8$ to $71.8$~ns per coefficient on 6 threads, a factor of 4.47. The gain comes from parallel execution. On one thread the level batch and scalar loop have comparable costs, about 321 and 325~ns per coefficient. Each thread count uses a fresh process, unpinned threads and the median of 9 warmed calls. 

\subsection{Exact workloads}
\paragraph{Exact values.} Exact workloads distinguish construction of the compact form from its display, expansion and equality tests. An \code{Exact(k)} value holds its rule and is written out only when it is shown. Constructing it therefore takes between $0.113$ and $0.160$~\textmu s for all 46 $6j$ symbols measured, with all spins equal to 1, 2 or 4 at levels up to 200, timed as the mean of 200 calls per sample and the median of 21 samples. Printing the value takes a median of $3.0$~\textmu s at $k=10$ with all spins equal to $1$, and $10.1$~\textmu s and $12.5$~\textmu s at $k=100$ and $k=200$ with all spins equal to $4$. Expanding it in $x$ with \code{xform} takes $15.8$, $87.7$ and $639$~\textmu s for the same inputs, with costs influenced by the field degree $\varphi(2h)/2$ and by the input expression. Two different values at $k=401$, with all spins equal to $100$ or with one of them $99$, are compared by \code{==} in $5.29$~\textmu s. Outside the timing, independent factorial sums in a cyclotomic field verify all 46 values exactly.

\paragraph{Zero decisions.} The zero query, the exact output and the numerical output agree on every one of the 40754 ordered admissible $6j$ inputs at levels $1$--$10$. Independent factorial sums in a cyclotomic field classify these inputs into 293 zeros and 40461 nonzeros, with selection-rule zeros excluded. The query (\code{iszero_at}) settles 245 of the zeros by structural identities and confirms the other 48 exactly, and \cref{tab:zerodecisions} gives the cost of each path. Constructing the \code{Exact(k)} value and testing it with \code{iszero} has group medians between $0.465$ and $2.33$~\textmu s, because the value is decided from its rule without being written out.

Further checks cover classical and weighted coupling coefficients and nearby levels, with both zeros and nonzeros, and exact confirmation rejects all 6 nonzero sums constructed to vanish modulo the screening prime. Nonzero values survive down to $5.7\times10^{-306}$ in \texttt{Float64}. A value near $2.5\times10^{-407}$ underflows in that type, while the exact query and a 512-bit evaluation both report it as nonzero.

\begin{table}[!htbp]
\centering\small
\caption{Zero-decision costs (\textmu s): median across $n$ inputs, selected with a fixed seed within each class, of 21 warmed samples of 20 calls each. The columns are the query \code{iszero_at}, construction of the \code{Exact(k)} output followed by \code{iszero}, and evaluation of the numerical coefficient.}
\label{tab:zerodecisions}

\begin{tabularx}{\columnwidth}{@{}X r r r r@{}}\toprule
Class & $n$ & Query & Exact & Numeric\\\midrule
Structural zero & 64 & 0.383 & 1.31 & 0.317\\
Other zero & 48 & 0.654 & 2.33 & 0.415\\
Nonzero & 64 & 0.198 & 0.465 & 0.102\\
\bottomrule\end{tabularx}

\end{table}

\paragraph{Setup.} Loading the package takes a median of $0.127$~s over 50 fresh processes using existing precompiled caches. Ten first-call operations, each measured in 5 separate processes after loading, take between $0.442$ and $12.5$~ms, including any remaining compilation and setup. Executable startup, installation and package precompilation are excluded. Building an uncached level table at $k=10000$ takes $0.501$~ms, and its live Julia storage occupies $1.22$~MiB.

\section{Discussion and outlook}\label{sec:limits}

QRecoupling connects the algebraic structure of finite factorial sums with their evaluation. A coefficient is stored as a factorial rule, and its cyclotomic exponents show which factors cancel within a product and which terms contribute at a unitary level. The cancellation that remains between terms is handled by adaptive arithmetic, and exact tests separate an algebraic zero from a small nonzero value. For coupling and recoupling coefficients, recurrences give a second representation that avoids this cancellation and computes related coefficients together. Numerical, exact and symbolic results come from one interface, and the outputs available depend on the coefficient and on the deformation.

These local $U_q(\mathfrak{sl}_2)$ data serve calculations in quantum geometry, topology and quantum many-body physics, where coefficients are rarely used one at a time. In such calculations, accurate individual coefficients do not by themselves guarantee an accurate observable. The numerical error estimates are conditional, and further cancellation can occur when coefficients enter a larger sum. A floating-point zero can also come from underflow, so a claim of algebraic vanishing needs an exact decision. Where it is supported, exact arithmetic removes both ambiguities (\cref{tab:targets}), at a cost that grows with integer size, field degree and expression complexity. 

The nearest extension is to higher $3nj$ symbols. The $9j$ and $15j$ symbols can be written as sums of products of $6j$ symbols, so the supported numerical and exact routes apply to their terms, and the new task is to control the accuracy of the outer sum. Higher-rank quantum groups would require new techniques, because their recoupling coefficients carry multiplicity labels and have no general formula of the kind used here. 

In tensor-network calculations, the coefficients supply the fusion data of anyonic chains and string nets~\cite{goldenchain,levinwen,pfeifer}. The sectors of \cref{sec:tensornetworks} make these data available to TensorKit at principal unitary levels and positive real $q$. In quantum gravity, recoupling coefficients determine the local amplitudes of three-dimensional state sums and enter the construction and coarse graining of spin-foam amplitudes~\cite{pr,tv,dona,gozzini,asante_tn,asante_rg,dittrich}. These applications call for both accurate numerical amplitudes and exact checks of the identities on which the models depend.

A natural next step is to carry the package's error estimates and exact decisions into larger contractions and state sums. In these calculations, admissibility determines the allowed channels and proven zeros identify vanishing contributions, while exact calculations on small systems provide references for numerical methods. Alongside these exact checks, tracking arithmetic error separately from truncation error would help distinguish numerical artefacts from physical effects. QRecoupling thus makes exact algebraic structure a practical computational resource for studying quantum matter and quantum geometry.

\section*{Acknowledgements}
The author thanks Hal Haggard for insightful discussions. The author acknowledges support from the Atlantic Association for Research in the Mathematical Sciences (AARMS) and the Natural Sciences and Engineering Research Council of Canada (NSERC).

\section*{Software availability}
QRecoupling.jl is distributed under the MIT license and requires Julia 1.10 or later. Source and issue tracking are at \url{https://github.com/sethkasante/QRecoupling.jl}. Documentation is at \url{https://sethkasante.github.io/QRecoupling.jl/}. This article describes version 0.6.4, archived at \href{https://doi.org/10.5281/zenodo.24027544}{doi:10.5281/zenodo.24027544}, and the measurements were made on that version. All versions are archived under \href{https://doi.org/10.5281/zenodo.19446095}{doi:10.5281/zenodo.19446095}.

\section*{Generative AI assistance}
Claude Code and OpenAI Codex assisted in software prototyping and refinement, testing and documentation review, and in editing parts of the manuscript. The author reviewed and edited the resulting code and text and takes responsibility for the software and article.

\appendix
\section{Error estimate for the plain pass}\label{app:error}
For the normalized sum $S_{\rm norm}=\sum_{j=0}^N\tau_j$, with $\tau_0=1$ and $\tau_{j+1}=r_j\tau_j$, write $\widehat\tau_j$ and $\widehat s_j$ for the computed terms and partial sums, initialized exactly at $1$. Then $\widehat S_{\rm norm}=\widehat s_N$. Assume at most $c$ relative error factors $1+\delta$, $|\delta|\le u$, per step used to form a term, including the ratio calculation. Here $u$ is unit roundoff. Assume no overflow or underflow, exact rescaling, and $cNu<1/2$. With $\gamma_m=mu/(1-mu)$, the absolute error satisfies
\begin{equation}\label{eq:bound}
\begin{split}
 |\widehat S_{\rm norm}-S_{\rm norm}|
 \le{}&\frac{cu}{(1-\gamma_{cN})^2}
          \sum_{j=1}^{N}j|\widehat\tau_j|\\
 &+\frac{u}{1-u}\sum_{j=1}^{N}|\widehat s_j|.
\end{split}
\end{equation}
The first sum accounts for forming the terms, and the second accounts for adding them. Both are accumulated during the pass.

For the derivation, standard product-error estimates~\cite{higham} give $|\widehat\tau_j-\tau_j|\le\gamma_{cj}|\tau_j|\le cju|\widehat\tau_j|/(1-\gamma_{cN})^2$. Each addition satisfies $\widehat s_j=(\widehat s_{j-1}+\widehat\tau_j)(1+\epsilon_j)$, $|\epsilon_j|\le u$, so
\[
 |\widehat s_j-\widehat s_{j-1}-\widehat\tau_j| \le\frac{u}{1-u}|\widehat s_j|.
\]
Summing these addition residuals telescopes. Adding the term errors proves \eqref{eq:bound}.

To estimate the error in the full coefficient, scale \eqref{eq:bound} by the absolute segment factor and include errors in the prefactor, first term, combination of segments and final rounding. The resulting computed estimate is $B$ in \cref{sec:escalation}. If $B$ bounds the absolute error and $B<|\widehat{\mathcal W}|$, the reverse triangle inequality gives
\[  \varepsilon_{\rm rel}\le\frac{B}{|\widehat{\mathcal W}|-B}. \]
Thus the acceptance test $B\le\eta|\widehat{\mathcal W}|$, $\eta<1$, implies $\varepsilon_{\rm rel}\le\eta/(1-\eta)$, with the relative error defined by \eqref{eq:relativeerror}.

This model covers the plain pass under the stated assumptions. It does not certify transcendental table values, rounding in the computation of $B$, or the separate fallback routes.

\end{document}